\documentclass[conference]{IEEEtran}

\usepackage[utf8]{inputenc}
\usepackage[english]{babel}

\usepackage{amsmath}
\usepackage{amssymb}
\usepackage{amsfonts}
\usepackage{amssymb}
\usepackage{amsfonts}
\usepackage{mathtools}
\usepackage{nicefrac}

\def\bbH{\mathbb{H}}
\def\bbI{\mathbb{I}}

\def\bbN{\mathbb{N}}

\def\bbP{\mathbb{P}}

\makeatletter
\let\IfStar\@ifstar 

\newcommand*{\eqdef}{\@ifstar{\eqdef@B}{\eqdef@A}}
\newcommand*{\eqdef@A}{\stackrel{{\scriptscriptstyle \mathrm{def}}}{\coloneqq}}
\newcommand*{\eqdef@B}{\coloneqq}

\newcommand*{\defeq}{\@ifstar{\defeq@B}{\defeq@A}}
\newcommand*{\defeq@A}{\stackrel{{\scriptscriptstyle \mathrm{def}}}{\eqqcolon}}
\newcommand*{\defeq@B}{\eqqcolon}

\def\@given@A{{\mkern2mu\mid\mkern1.5mu}}
\def\@given@B{\middle|}
\let\given\@given@A
\newcommand*{\useGiven@B}[1]{\def\given{\@given@B}#1\def\given{\@given@A}}

\mathchardef\equals=\mathcode`=
\def\equals@A{{\mkern0mu\equals\mkern1mu}}
\def\equals@B{{\mkern0.5mu\equals\mkern1.5mu}}
\def\useEquals@A#1{\begingroup\lccode`~=`=\lowercase{\endgroup\def~}{\equals@A}\mathcode`=="8000{#1}\begingroup\lccode`~=`=\lowercase{\endgroup\def~}{\equals}}
\def\useEquals@B#1{\begingroup\lccode`~=`=\lowercase{\endgroup\def~}{\equals@B}\mathcode`=="8000{#1}\begingroup\lccode`~=`=\lowercase{\endgroup\def~}{\equals}}

\newcommand*{\style@A}[1]{\useEquals@A{#1}}
\newcommand*{\style@B}[1]{\useEquals@B{\useGiven@B{#1}}}

\newcommand*{\set}{\@ifstar{\set@B}{\set@A}}
\newcommand*{\set@A}[1]{\style@A{\{#1\}}}
\newcommand*{\set@B}[1]{\style@B{\left\{#1\right\}}}

\newcommand*{\abs}{\@ifstar{\abs@B}{\abs@A}}
\newcommand*{\abs@A}[1]{\style@A{|#1|}}
\newcommand*{\abs@B}[1]{\style@B{\left|#1\right|}}

\newcommand*{\norm}{\@ifstar{\norm@B}{\norm@A}}
\newcommand*{\norm@A}[1]{\style@A{\Vert#1\Vert}}
\newcommand*{\norm@B}[1]{\style@B{\left\lVert\{#1\right\rVert}}

\newcommand*{\lr}{\@ifstar{\lr@B}{\lr@A}}
\newcommand*{\lr@A}[1]{\style@A{(#1)}}
\newcommand*{\lr@B}[1]{\style@B{\left(#1\right)}}

\newcommand*{\LR}{\@ifstar{\LR@B}{\LR@A}}
\newcommand*{\LR@A}[1]{\style@A{[#1]}}
\newcommand*{\LR@B}[1]{\style@B{\left[#1\right]}}

\newcommand*{\vecList}{\@ifstar{\vecList@B}{\vecList@A}}
\newcommand*{\vecList@A}[1]{{\LR{#1}}^\top}
\newcommand*{\vecList@B}[1]{{\LR*{#1}}^\top}

\makeatother

\usepackage{xstring}
\newcommand*{\prob}[1]{\IfStrEq{#1}{_}{\ProbWithSub}{\bbP\lr{#1}}}
\newcommand*{\ProbWithSub}[2]{{\operatorname*{\bbP}_{#1}(#2)}}

\usepackage{breqn}
\usepackage{pgfplots} \pgfplotsset{compat=1.17}
\usepackage{hyperref}
\usepackage{graphicx}

\usepackage{tikz}
\usetikzlibrary{arrows,shapes}
\usetikzlibrary{arrows.meta}
\usepackage[justification=raggedright,singlelinecheck=false]{subcaption}
\usepackage{float}
\usepackage{xcolor}
\usepackage{amsthm}
\usepackage{thmtools}
\usepackage{changepage} 
\usetikzlibrary{patterns}

\usepackage{pgfplotstable}

\pgfarrowsdeclarecombine*{ring and diamond}{ring and diamond}{open diamond}{open diamond}{o}{o}


\def\citeRepository{\cite{githubRepository}}
\def\keywords#1{\begin{IEEEkeywords}#1\end{IEEEkeywords}}
\hypersetup{
colorlinks=true,
linkcolor=yellow!25!red!80!black,
filecolor=blue,
citecolor=black!60!yellow!25!green,      
urlcolor=black!50!cyan,
}

\def\theLeakage{Rényi-min leakage} 
\def\TheLeakage{Rényi-min leakage} 
\def\theLeakageOptimal{Rényi-min optimal} 
\def\theLeakageThen#1{Rényi-then-#1}

\def\Mb{\texttt{Mb}}
\def\PrpRe{\texttt{PrpRe}}
\def\PopRe{\texttt{PopRe}}
\def\PrpReBa{\texttt{PrpReBa}}
\def\PopReSh{\texttt{PopReSh}}
\def\PopSh{\texttt{PopSh}}

\numberwithin{equation}{section}
\numberwithin{figure}{section}

\theoremstyle{plain}\newtheorem{theorem}{Theorem}
\theoremstyle{plain}\newtheorem{proposition}[theorem]{Proposition}
\theoremstyle{plain}
\theoremstyle{plain}\newtheorem{corollary}[theorem]{Corollary}
\theoremstyle{definition}
\theoremstyle{definition}\newtheorem{specification}[theorem]{Problem specification}
\theoremstyle{remark}\newtheorem{observation}[theorem]{Observation}
\theoremstyle{remark}

\begin{document}

\title{Minimizing Information Leakage under Padding Constraints}

\author{%
\IEEEauthorblockN{Sebastian Simon}
\IEEEauthorblockA{LIX\\
École Polytechnique\\
Paris, France}
\and
\IEEEauthorblockN{Cezara Petrui}
\IEEEauthorblockA{LIX\\ 
École Polytechnique\\
Paris, France}
\and
\IEEEauthorblockN{Carlos Pinzón}
\IEEEauthorblockA{Inria and LIX\\
École Polytechnique\\
Paris, France}
\and
\IEEEauthorblockN{Catuscia Palamidessi}
\IEEEauthorblockA{Inria and LIX\\
École Polytechnique\\
Paris, France}
}


\maketitle

\begin{abstract}

An attacker can gain information of a user by analyzing its network traffic.
The size of transferred data leaks information about the file being transferred or the service being used, and this is particularly revealing when the attacker has background knowledge about the files or services available for transfer.
To prevent this, servers may pad their files using a \emph{padding scheme}, changing the file sizes and preventing anyone from guessing their identity uniquely.
This work focuses on finding optimal padding schemes that keep a balance between privacy and the costs of bandwidth increase.
We consider \theLeakage{} as our main measure for privacy, since it is directly related with the success of a simple attacker, and compare our algorithms with an existing solution that minimizes Shannon leakage.
We provide improvements to our algorithms in order to optimize average total padding and Shannon leakage while minimizing \theLeakage{}.
Moreover, our algorithms are designed to handle a more general and important scenario in which multiple servers wish to compute padding schemes in a way that protects the servers' identity in addition to the identity of the files.

\end{abstract}




\section{Introduction}

This paper focuses on the problem of minimizing the information that an adversary can obtain when he analyzes the network traffic of a user.

Network traffic analysis is a powerful tool that is needed for optimizing the routing and transmission speed of data in telecommunications networks~\cite{welzl2005network,pouzols2011mining}, as well as for detecting anomalies~\cite{bhuyan2017network} and possibly some types of attacks.
In addition to this, traffic analysis is, on its own, a subject which sparks various research topics. For instance, there has been discovered evidence of self-similar patterns in network traffic~\cite{park2000self}.
Nonetheless, it can also be used to infer users' demographics~\cite{li2016demographics} and the types of web services they are using~\cite{hajjar2015network,callado2009survey}.
Therefore, network traffic analysis is an active and important research topic in telecommunications, networking, privacy, computer security, as well as in mathematics.


Since the vast majority of modern traffic is encrypted, the research on traffic analysis has increasingly been focused towards the analysis of encrypted traffic, in which the adversary willing to gain information can not observe the content of the messages, but can make inferences based on their timing and sizes.
Although data-encryption provides a strong limitation for the attacker, it is possible, with a substantial degree of accuracy, to detect information about the activity of the users, such as what site is being browsed~\cite{siby2018dns,cherubin2017website} or the format of a streamed video~\cite{wampler2015information}.


Moreover, as pointed out in~\cite{wright2009traffic}, packet sizes and timing are essential tools that can reveal the language of a VoIP call~\cite{wright2007language}, passwords in secure shell logins~\cite{song2001timing}, or even web browsing habits \cite{sun2002statistical, liberatore2006inferring}.
The danger of these inferences is highly increased when the adversary combines several of them, possibly also with additional background knowledge about the user or about application standards, e.g. if Google Meet and Skype use different video formats by default, the attacker would be able to identify which application is being used.

Padding messages is a well-known technique for reducing the quality of the inferences that the traffic observer can do.
It increases the sizes of the files being transferred across a network in order to make them harder to recognize.
More precisely, padding reduces the probability that the attacker infers correctly the file being transferred.
But it also increases the bandwidth use of the network.
Therefore, it is necessary to design padding-schemes that achieve an optimal trade-off between reducing bandwidth overhead and its costs while maintaining a high level of privacy.

We propose different algorithms that minimize \theLeakage{} subject to three different types of bandwidth constraints that occur naturally in practice.
We prove the correctness of all the algorithms described in the paper, test them against brute-force implementations over small datasets, and compare them to other methods in the literature over a large dataset.


\section{Related Works}

The need for padding has been identified already in~\cite{siby2018dns} and solved for specific applications in~\cite{wright2009traffic} and \cite{reed2021optimally}.

This paper is strongly related with the work of Reed and Reiter~\cite{reed2021optimally}, in which the authors propose three padding algorithms, PRP\_Shannon, POP\_Shannon (referred to as \PopSh{} in this document) and PwoD (padding without a distribution), for finding padding schemes that minimize Shannon leakage under different bandwidth constraints.
They formalize the padding problem as a server that pads its files before or while serving them to clients.

We extend~\cite{reed2021optimally} in two ways.
First, we take also into consideration the problem of \emph{server identity protection}, in which several servers coordinate themselves before deploying their optimal padding-schemes such as to protect their identities in addition to the individual identities of the files.
We adequate their POP algorithm to minimize Shannon leakage for this new scenario.
Secondly, we redesign all their algorithms to find padding functions that minimize \theLeakage{}~\cite{Smith:09:FOSSACS} instead of Shannon leakage, and we provide some heuristics for \theLeakageThen{Shannon} and \theLeakageThen{bandwidth} minimization, i.e. finding among all schemes that minimize \theLeakage{}, the one that minimize Shannon leakage or bandwidth respectively.


We believe that \theLeakage{} is a better metric for privacy than Shannon leakage for the the file padding problem.
We justify this claim in this paper by pointing out that \theLeakage{} is directly related via a bijective function with the probability of an attacker's best inference to be true. This correspondence can also be interpreted as an instance (to network traffic analysis) of the general operational interpretation of \theLeakage{} in terms of one-try attacks~\cite{Smith:09:FOSSACS}.
For Shannon, which is more related to information theoretic applications like compression, the attackers' success and the leakage are highly correlated, however, they are not directly related via a linear function.
This fact is corroborated experimentally in~\cite{reed2021optimally}, as well as in our paper.

In general, in the privacy community, metrics are better described in terms of their associated attacker rather than their information-theoretic properties~\cite{Alvim:12:CSF,romanelli2020machine}.
For this particular application, the direct pragmatic connection between \theLeakage{} and a simple adversary success makes it very appealing.
The same argument is used in~\cite{cherubin2017bayes}, whose privacy measure is closely related with ours.








\def\size#1{{|#1|}}

\section{Problem formalization}

Generally speaking, the objective is to determine how much to pad each of the files stored in a server to minimize the information gained by a potential adversary analyzing network traffic.

The files in the server store are denoted as $E=\{e_1,e_2, \dots , e_n \}$.
We assume the files in $E$ to be sorted non-decreasingly by their sizes $\size{e_i}$, and to be accessed with frequency $p_i\in[0,1]$, where $\sum_{i=1}^n p_i=1$.
To denote a random file from the server, we use the random variable $X$ that takes values in $E$ with probabilities $\prob{X = e_i}\eqdef p_i$.
We define the set of file sizes $S\eqdef\set{|e| | e\in E}$ and enumerate its unique elements in increasing order as $S = \{s_1, s_2, \dots , s_m\}$.
Note that $m \leq n$, since there may be different files with the same size. 

Given a fixed multiplicative constraint $c\geq 1$, a \emph{padding function}, or padding scheme, is a (possibly randomized) function $f :E \rightarrow \mathbb{N}$ that satisfies, for all $i \in [1..n]$, $\size{e_i} \leq f(e_i) \leq c \cdot \size{e_i}$ (with probability $1$).
The padding function will determine the size of the files in the server after padding them, and the double inequality represents, on the one hand, the fact that files can only be padded to a size greater than the initial one:
\begin{equation} \label{constraint_1}
    \prob{f(e)\geq\size{e}} = 1,
\end{equation}
and on the other hand, the constraint of avoiding unexpectedly large paddings:
\begin{equation} \label{constraint_2}
    \prob{f(e)\leq c \cdot \size{e}} = 1.
\end{equation}
The latter constraint guarantees both for the server and the client that the bandwidth used is not excessive.
It is multiplicative instead of additive because we consider typical cases in which the largest files are accessed less frequently than the smallest ones, hence, it is desirable to limit more severely the padding added to smaller files.
Also, it might be randomized in case the server has the infrastructure to pad files independently for each transfer request.
Moreover, in the favor of clients, the constraint is enforced for each individual file instead of a single global expectation of bandwidth increase because if a client accesses just one or two files, it is of his interest that the bandwidth is bounded for each of those two files.

The objective is to find a padding function $f:E\to\bbN$ that minimizes a given metric of information leakage $\bbI(\size{X}, f(X))$, while respecting the constraints $\size{e_i}\leq f(e_i) \leq c \cdot \size{e_i}$, for each $i\in\{1,...,n\}$.

In the rest of this section, we describe all the variations of the problem by considering two different leakage functions from the literature, as well as several additional types of constraints that may arise in real-world applications.

\subsection{Leakage functions}

In the literature there are multiple options for quantifying the leakage $\bbI\lr{\size{X}, f(X)}$ of a padding function $f:E \rightarrow \mathbb{N}$.
We will further use the notation $\bbI(f)\eqdef\bbI\lr{\size{X}, f(X)}$ to reduce verbosity since the random variable $X$ is fixed.

In~\cite{reed2021optimally}, they use Shannon mutual information as a leakage measure.
Shannon leakage is based on Shannon entropy $\bbH$ and is given by:
\begin{equation} \label{Shannon_leakage}
    \bbI_1\lr{\size{X}, f(X)} \eqdef \bbH\lr{\size{X}}-\bbH\lr{\size{X}\given f(X)}
\end{equation}

Shannon leakage is a particular case ($\alpha=1$) of a family of leakages $\bbI_{\alpha}$ based on $\alpha$-Rényi entropy.
A particular case of interest for privacy applications is what we call \emph{\theLeakage{}} ($\alpha=\infty$), whose derivation and importance is highlighted in~\cite{palamidessi2018feature} and~\cite{smith2009foundations}:
\begin{equation} \label{Renyi_leakage}
    \mathbb{I}_{\infty}(\size{X},f(X))=\mathbb{H}_{\infty}(\size{X})-\mathbb{H}_{\infty}(\size{X} \mid f(X)).
\end{equation}

Both Shannon leakage and \theLeakage{} are directly related (in one-to-one relation) to the probability of success of an attacker, but the specific assumptions about the attacker and his objective are different.
In the case of Shannon leakage, the attacker is assumed to have the power to perform set-queries of the type "is the secret in \emph{this} set?", and his objective is to guess the secret using the minimal number of queries.
For \theLeakage{}, the attackers' power is reduced to a one-try attack, i.e. he can only make one query of the type "is \emph{this} the secret?", and his objective is to maximize the probability of being correct.
In both cases, to protect against worst-case scenarios, the attacker is assumed to know the server's files' sizes and frequencies, as well as the padding-scheme used by the server.
This knowledge may arise in practice if the server hosts public files and the attacker has had time to request each of the files, possibly multiple times in case of a randomized padding scheme.
Notice that the fact that the server files are public, does not reduce at all the need for privacy.
For instance, although the files of an adult video site are public, the users are very interested that no adversarial third party knows exactly which video is being watched.
Therefore, it is clear from the point of view of attackers, that \theLeakage{} is superior to Shannon leakage as a measure of privacy for this particular application.

We conclude this section by proving an important property of padding functions that minimize \theLeakage{}.
The same holds true for paddings that minimize Shannon leakage, as shown in~\cite{reed2021optimally}.

\begin{proposition} \label{proposition_1}
Let $\mathbb{I}$ denote \theLeakage{}, $S$ the set of sizes of the files and $f:E \rightarrow \mathbb{N}$ a padding-scheme. Then, there is a padding-scheme $f^*:E \rightarrow S$ such that $\mathbb{I}(f^*) \leq \mathbb{I}(f)$ and $\mathbb{I}_{\infty}(f^*) \leq \mathbb{I}_{\infty}(f)$.
\end{proposition} 
\begin{proof}
We consider that the files are sorted in non-decreasing order with respect to their sizes, denoted by the set $S = \{s_1, s_2, \dots, s_m\}$.
We have $X$ the random variable associated to $E$ and denote the random variable $Y=f(X)$.
We define $g: \mathbb{N} \rightarrow S$ to be the function satisfying $g(f(e_i)) = \max \{s \in S |s \leq f(e_i)\}$.
Denote by $Z$ the random variable $(g \circ f)(X)$ equipped with the probability space of $E$.
We use the Data Processing Inequality for $\size{X}, Y, Z$ as stated in Theorem 6.2 of \cite{m2012measuring} and Theorem 8 of \cite{mciver2014abstract}, and remarked also in Theorem 5.1 of \cite{smith2015recent}.
$Z$ only depends on $Y$ and is conditionally independent of $\size{X}$ because $Z$ is a deterministic function of $Y$.
Since this holds for any $g$-leakage function, we have $ \mathbb{I}(\size{X},f(X))  \geq \mathbb{I}(\size{X}, (g \circ f)(X))$ and $ \mathbb{I_{\infty}}(\size{X},f(X))  \geq \mathbb{I_{\infty}}(\size{X}, (g \circ f)(X))$.
Hence, we proved that there exists a function $f^*=g \circ f:E \rightarrow S$ such that it has a lower leakage than $f$.
\end{proof}

\begin{corollary} \label{corollary_1}
A padding function that has minimal leakage must pad each file to the size of another file in the initial set.
\end{corollary}

\subsection{Variations of the padding problem}

The padding problem can have additional restrictions that arise commonly in practice, giving rise to the following variations, of which the first two are also considered in~\cite{reed2021optimally}.

\begin{enumerate}
\item
\textbf{POP (Per-object-padding)} 

In this variation, the server plans to pad the files only once and forever.
This saves the server from costs and complexity associated to repeated padding.
Hence, our objective in this variation is to compute a deterministic padding function $f$.

\item
\textbf{PRP (Per-request-padding)}

In this variation, we assume no additional constraints to the padding function others than Equations \ref{constraint_1} and \ref{constraint_2}. 
This means that the store can pad the files before every transaction.

\item
\textbf{Server identity protection}

Note that although the generic problem has no restriction on the output sizes $f(X)$, Corollary~\ref{corollary_1} implies that the optimal padding is achieved when $f(X)=S$, meaning that the files of the store can only be padded to sizes of other files in the set $E$.
This fact has a critical implication regarding the protection of server identities.
For instance, assume that an attacker who knows that server $A$ and $B$ have file sizes $20, 22, 24$ and $21, 23, 25$ respectively, and he observes a (padded) file of size $22$ being transferred to a user; then he does not know if the file is that of size $20$ or $22$, due to padding, but he knows certainly that the user accessed server $A$.
In order to protect the server identity, we suggest a standardized protocol among the servers.
This protocol involves having all the servers padding the files to the same pre-determined set of sizes $Z$ such that $f(X) \subset Z$.
In practice, $Z$ can be a set of rounded sizes such as $\{5\Mb, 10\Mb, 20\Mb, 30\Mb, \dots \}$.
We keep the padding constant $c$ and conditions \eqref{constraint_1} and \eqref{constraint_2}.
Additionally, we add the constraint that, $\forall {i \in [1..n]}$:
\begin{equation} \label{Condition_extra}
    \set{z \in Z \given z \geq \size{e_i},\; z \leq c \cdot \size{e_i}} \neq \emptyset 
\end{equation}
meaning that the set $Z$ has enough possible padding sizes to be consistent with \eqref{constraint_2}.
Moreover, the protocol aforementioned for protecting the servers holds not only for Per-Object-Padding, but also for Per-Request-Padding, i.e. it does not influence in any way if the padding function $f$ is deterministic or probabilistic, since we can apply the same procedure for both cases.
\end{enumerate}

\subsection{Generalization}

All the aforementioned variations of the padding problem can be written as instances of a single input-output problem specification.
The specification requires to standardize the bandwidth constraints as pairs $[l_i,r_i]$, meaning that file $e_i$ can only be padded to sizes $\set{z_{l_i}, ..., z_{r_i}}$, for some fixed increasing sequence of output values $Z=\set{z_1, ..., z_{|Z|}}$.
Naturally, these pairs must satisfy some minimal properties to reflect a possible instance of the padding problem.
Namely, we say a sequence $\set{\LR{l_i, r_i}}_{i=1}^n$ is \emph{a valid sequence of constraints} for $Z$ if $l_i\leq r_i$ holds for all $i \in [1..n]$, and the sequences $\set{l_i}_{i=1}^{n}$ and $\set{r_i}_{i=1}^{n}$ are non-decreasing and satisfy $l_1 = 1$ and $l_n = r_n = |Z|$.

We now state the problem specification and proceed to explain how it captures all variations of the padding problem.
All the algorithms in following sections follow this specification.

\begin{specification}\label{specification}\,\\
\textbf{Input}:
\begin{itemize}
  \item[]%
  A set $E$ of $n$ files $\set{e_i \given i \in [1..n]}$ with sizes $\size{e_i}$ and frequencies $p_i$;
  a set $Z$ of possible output sizes ($Z=S$ unless otherwise specified) with a valid sequence of constraints $\set{[l_i, r_i]}_{i=1}^n$;
  and a leakage function $\bbI$ that measures the informational leakage of padding-schemes.
\end{itemize}
\textbf{Desired output}:
\begin{itemize}
  \item[] %
  A padding function $f:E\to Z$ (deterministic for POP, possibly randomized for PRP) that minimizes the information leakage $\bbI(\size{X}, f(X))$, where $\prob{X=e_i}\eqdef p_i$ for each $i\in\{1,...,n\}$.
\end{itemize}
\end{specification}

If $f$ is deterministic, it can be encoded as a list $\lr{f(e_i)}_{i=1}^n$, and in the general probabilistic case as the channel-matrix $(p_{ij})$ of size $n \times |Z|$ between the secrets $X$ and the observables $f(X)$.

Firstly, we argue that the specification generalizes the variations POP and PRP.
In both cases, we invoke Corollary~\ref{corollary_1} which mandates $Z=S$, hence the the padding function will have signature $f:E \rightarrow S$.
Condition \eqref{constraint_1} is attained by choosing $l_i\eqdef\min\set{j \in [1..m]\given \size{e_j} \geq \size{e_i}}$ and $r_i \eqdef \max \set{j \in [1..m] \given s_j \leq c \cdot \size{e_i}}$, thus specifying the padding restriction individually for each file.
Note that this also generalizes POP and PRP, allowing the store to have special padding restrictions for each specific file rather than a universal constant $c$, e.g. a list $c_i$ for $i\in [1..n]$.

Secondly, to deal with the server identity protection case, we start with a padding-scheme $f:E \rightarrow Z$, the probability distribution $\mathbb{P}$, and the set $S$, with the purpose of adapting these three to the generalized problem.
For each $i \in [1..n]$, let $L^*_i=\min\set{z \in Z \given z \geq \size{e_i}}$, i.e. the smallest file size in $Z$ that file $e_i$ can be padded to.
Define $e^*_i$ to be a file with size $L^*_i$.
Due to \eqref{Condition_extra}, $L^*_i$ exists, $\forall {i \in [1..n]}$, and consider the multiset $S^*:=\set{L^*_i \given i \in [1..n]}$.
Because each file is mapped injectively to an element of $S^*$, we consider that we pad each element to the closest one in $Z$ and then apply the padding problem for the new files.
Using this and Corollary \ref{corollary_1}, we reduce the problem to finding a padding-scheme $f^*:E^* \rightarrow S^*$, and a sequence of restrictions $[l_i, r_i], i \in [1..n]$ that satisfy the relation $r_i:= \max \set{j \given L^*_j < c \cdot L^*_i}$, where $L^*_j$ is the size of the file in $E^*$.
Note that the sizes of the files are given in the three variations of the problem, but not in the generalization, as these can be reduced to the restrictions sequence.

\section{Algorithms (and proofs)}

In this section, we give algorithms that minimize the \theLeakage{} as defined in \eqref{Renyi_leakage} for the POP and PRP cases, namely \PopRe{} and \PrpRe{}, which contrast those for Shannon mutual information minimization, found in the paper \cite{reed2021optimally}.
The complexities of these algorithms are summarized in Table~\ref{table_complexities}.

\begin{table}[!ht]
\begin{tabular}{|c|c|c|}
  \hline
  Algorithm & Minimizes & WC Runtime complexity\\
  \hline
  \PopRe{} & \theLeakage{} & $O(n^2\, \bar m)$\\
  \PrpRe{} & \theLeakage{} & $O(n\, {\bar m})$\\
  \PopSh{} & Shannon leakage & $O(n\, {\bar m})$\\
  \texttt{ShannonPRP} & Shannon leakage & $O(\textsc{iters}\cdot n\,m)$ \\
  \hline
\end{tabular}
\def\theCaption{
  Complexities.
  Here, $m\leq n$ is the number of different file sizes and ${\bar m} \eqdef (\nicefrac{1}{n})\sum_{i=1}^n r_i-l_i+1$ is the average number of available choices for each file, thus, $O(n \, {\bar m})\leq O(n \, m)$.
}
\caption{\protect\theCaption{}} 
\label{table_complexities}
\end{table}
Notice that the algorithm \texttt{ShannonPRP} is an approximation algorithm and has a runtime complexity that depends on the degree of accuracy imposed by the user and the limit number of iterations $\textsc{iters}$ allowed.
Also, the complexities of the dynamic programming algorithms appear overestimated in the theoretical worst-case scenario when compared to the actual implementations.
For instance, although \PopRe{} has two parameters varying in $[1..n]$, not all combinations need to be calculated in a top-down implementation.


\subsection{\PopRe{}}

In this section we develop the algorithm that minimizes \theLeakage{} in the POP variation following Specification~\ref{specification}.
Before starting with the algorithm, we will prove Observation \ref{observation_1}, which will be used as the main update of the entries of the channel-matrix.

Let $f:E \rightarrow S$ (by the Corollary \ref{corollary_1}) be an optimal padding-scheme.
We want to minimize
    $$\mathbb{I}_{\infty}(\size{X},f(X))=\mathbb{H}_{\infty}(\size{X})-\mathbb{H}_{\infty}(\size{X}\mid f(X)),$$
but since $\mathbb{H}_{\infty}(\size{X})$ is constant in regards to the padding-scheme, maximizing $\mathbb{H}_{\infty}(\size{X}\mid f(X))$ is a sufficient condition.
As introduced in \cite{palamidessi2018feature}, we have:
    $$ \mathbb{H}_{\infty}(\size{X}\mid f(X))= -\log_2 \sum_{j \in [1..m]} \max_{i \in [1..n]} (p_i\cdot \mathbb{P}(f(e_i)=s_j)).$$
Given that $X \rightarrow f(X)$ can be seen as Markov Chain, it is natural to denote $\mathbb{P}(f(e_i)=s_j)=p_{ij}$.
Note that we are in the case of per-object-padding, so $p_{ij} \in \set{0,1} \forall i\in [1..n] \text{ and } j \in [1..m]$.
Because the logarithmic function is increasing, the problem reduces to minimizing:
\begin{equation} \label{sum_min}
\sum_{j \in [1..m]} \max_{i \in [1..n]} (p_i\cdot p_{ij}).
\end{equation}

\begin{figure}[!ht]
    \def\Rcm{0.55cm}
  \def\Ccm{0.4cm}
  \begin{tikzpicture}[x={(0,-\Ccm)}, y={(\Rcm, 0)}, every node/.style={minimum width=\Rcm, minimum height=\Ccm, outer sep=0pt, anchor=north west}]
    \node[fill=gray, minimum width=1*\Rcm] at (2,0) {};
    \node[fill=gray, minimum width=2*\Rcm] at (3,0) {};
    \node[fill=gray, minimum width=3*\Rcm] at (4,0) {};
    \node[fill=gray, minimum width=4*\Rcm] at (5,0) {};
    \node[fill=gray, minimum width=4*\Rcm] at (0,2) {};
    \node[fill=gray, minimum width=2*\Rcm] at (1,4) {};
    \node[fill=gray, minimum width=1*\Rcm] at (2,5) {};
    \node[fill=gray, minimum width=1*\Rcm] at (3,5) {};
    \node at (-0.75,-0.95) {$\vdots$};
    \node at (1,-0.95) {$e_{9}$};
    \node at (2,-1.20) {$e_{10}$};
    \node at (3,-1.20) {$\mathbf{e_{11}}$};
    \node at (4,-1.20) {$e_{12}$};
    \node at (4.3,-0.95) {$\vdots$};
    \node[xshift=-0.2*\Ccm] at (-0.90,0) {$\cdots$};
    \node[xshift=-0.2*\Ccm] at (-0.90,1) {$s_{10}$};
    \node[xshift=-0.2*\Ccm] at (-0.90,2) {$s_{11}$};
    \node[xshift=-0.2*\Ccm] at (-0.90,3) {$\mathbf{s_{12}}$};
    \node[xshift=-0.2*\Ccm] at (-0.90,4) {$s_{13}$};
    \node[xshift=-0.2*\Ccm] at (-0.90,5) {$\cdots$};
    \node at (0,0) {$1$};
    \node at (1,0) {$1$};
    \node at (2,2) {$1$};
    \node at (3,3) {$\mathbf{1}$};
    \node at (4,5) {$1$};
    \node at (5,5) {$1$};
    \node[color=gray] at (0,1) {$0$};
    \node[color=gray] at (1,1) {$0$};
    \node[color=gray] at (1,2) {$0$};
    \node[color=gray] at (1,3) {$0$};
    \node[color=gray] at (2,1) {$0$};
    \node[color=gray] at (2,3) {$0$};
    \node[color=gray] at (2,4) {$0$};
    \node[color=gray] at (3,2) {$0$};
    \node[color=gray] at (3,4) {$0$};
    \node[color=gray] at (4,3) {$0$};
    \node[color=gray] at (4,4) {$0$};
    \node[color=gray] at (5,4) {$0$};
    \draw[xstep=\Rcm, ystep=\Ccm,color=black,xshift=-0.5\pgflinewidth,yshift=0.5\pgflinewidth] (0,0) grid (6, 6);
  \end{tikzpicture}
  \begin{tikzpicture}[x={(0,-\Ccm)}, y={(\Rcm, 0)}, every node/.style={minimum width=\Rcm, minimum height=\Ccm, outer sep=0pt, anchor=north west}]
    \node[fill=gray, minimum width=1*\Rcm] at (2,0) {};
    \node[fill=gray, minimum width=2*\Rcm] at (3,0) {};
    \node[fill=gray, minimum width=3*\Rcm] at (4,0) {};
    \node[fill=gray, minimum width=4*\Rcm] at (5,0) {};
    \node[fill=gray, minimum width=4*\Rcm] at (0,2) {};
    \node[fill=gray, minimum width=2*\Rcm] at (1,4) {};
    \node[fill=gray, minimum width=1*\Rcm] at (2,5) {};
    \node[fill=gray, minimum width=1*\Rcm] at (3,5) {};
    \node at (-0.75,-0.95) {$\vdots$};
    \node at (1,-0.95) {$e_{9}$};
    \node at (2,-1.20) {$e_{10}$};
    \node at (3,-1.20) {$\mathbf{e_{11}}$};
    \node at (4,-1.20) {$e_{12}$};
    \node at (4.3,-0.95) {$\vdots$};
    \node[xshift=-0.2*\Ccm] at (-0.90,0) {$\cdots$};
    \node[xshift=-0.2*\Ccm] at (-0.90,1) {$s_{10}$};
    \node[xshift=-0.2*\Ccm] at (-0.90,2) {$s_{11}$};
    \node[xshift=-0.2*\Ccm] at (-0.90,3) {$\mathbf{s_{12}}$};
    \node[xshift=-0.2*\Ccm] at (-0.90,4) {$s_{13}$};
    \node[xshift=-0.2*\Ccm] at (-0.90,5) {$\cdots$};
    \node at (0,0) {$1$};
    \node at (1,3) {$\mathbf{1}$};
    \node at (2,3) {$\mathbf{1}$};
    \node at (3,3) {$\mathbf{1}$};
    \node at (4,3) {$\mathbf{1}$};
    \node at (5,5) {$1$};
    \node[color=gray] at (0,1) {$0$};
    \node[color=gray] at (1,0) {$0$};
    \node[color=gray] at (1,1) {$0$};
    \node[color=gray] at (1,2) {$0$};
    \node[color=gray] at (2,1) {$0$};
    \node[color=gray] at (2,2) {$0$};
    \node[color=gray] at (2,4) {$0$};
    \node[color=gray] at (3,2) {$0$};
    \node[color=gray] at (3,4) {$0$};
    \node[color=gray] at (4,4) {$0$};
    \node[color=gray] at (4,5) {$0$};
    \node[color=gray] at (5,4) {$0$};
    \draw[xstep=\Rcm, ystep=\Ccm,color=black,xshift=-0.5\pgflinewidth,yshift=0.5\pgflinewidth] (0,0) grid (6, 6);
  \end{tikzpicture}
  \def\theCaption{
    Visualization for Observation~\ref{observation_1}.
    The file with maximal frequency is $e_{11}$, and the left and right padding-schemes are respectively $f$ and $f^*$.
    If the left one is optimal, the right one must be as well.
    }
  \caption{\protect\theCaption{}}
  \label{Figure_1}
\end{figure}

\begin{observation}\label{observation_1}
Let $f$ be a \theLeakageOptimal{} padding-scheme and $e_i$ be the file with the highest associated frequency $p_i$, and assume that $p_{ij}=1$ for some $j \in [1..m]$. Then there exists a padding-scheme $f^*$ with the same \theLeakage{} such that  $p_{kj}=1$ for all $k \in [1..n]$ such that $j\in[l_k..r_k] $.
\end{observation}

\begin{proof}
We consider the padding-scheme $f$ to be represented as the channel-matrix between the secrets and the observables. Let $P$ be the $n \times m$ matrix which contains on every entry $(a,b)$ with $a\in [1..n], b \in [1..m]$ the probability $p_{ab}$. When we want to minimize \eqref{sum_min} we sum over each column of the matrix $P$. In particular, on the column $j$ we have $\max_{a \in [1..n]} (p_a\cdot p_{aj}) = p_i$ since $p_i$ is the highest frequency among the frequencies of the files and $p_{ij}=1$. Now, let us consider the padding-scheme $f^*$ such that with the associated matrix $P^*$, moving, on the same line, every 1 that we can to column j:
$$
p^*_{ab} = 
    \begin{cases}
        p_{ab} &  \text{ if }   b \neq j \text{ and } a\in[1..n] \text{ such that }
        j\not \in [l_a..r_a]\\
        1 & \text{ if }    b =j   \text{ and }   a \in [1..n]   \text{ such that }  j\in [l_a..r_a]  \\
        0 & \text{otherwise}
    \end{cases}
$$

On the column $j$ of the matrix $P^*$ we will still have $\max_{a \in[1..n]} (p_a \cdot p^*_{aj}) = p_i$ because the padding-scheme $f^*$ preserves the maximum on column $j$. Moreover, on the rest of the columns, the maximum either decreases or stays the same since we created more entries $p^*_{ab} = 0$, which means that the product $p_a \cdot p^*_{ab} = 0$. However, we chose $f$ to be the \theLeakageOptimal{} padding-scheme and with the observations above, $f$ and $f^*$ give the same leakage.
\end{proof}

Figure~
\ref{Figure_1} depicts an example of a sub-matrix of $P$ as described in Observation \ref{observation_1}. In the figure, we have exactly one entry equal to $1$ in each line because the channel-matrix is stochastic and we are in the POP case. Additionally, the quantity in \eqref{sum_min} represents the sum of the maximum over columns where each $1$ counts for the frequency of the file. Then, the update does not increase the \eqref{sum_min} because the $1$ with maximal frequency dominates its column, and moving all possible $1$'s above or below it does not increase \theLeakage{}.

Using Observation \eqref{observation_1} we can divide the padding problem into sub-problems that minimize \eqref{sum_min} and leverage dynamic programming: $\forall a \leq b \in [1..n]$, we define

$$
D[a][b]= \min_{\text{P channel matrix}} {\sum_{j \in [1..m]} \max_{i \in [a+1..b]} (p_i\cdot p_{ij}}),
$$
i.e. $D[a][b]$ gives the minimal leakage for the sub-problem that pads files from $e_{a+1}$ to $e_b$, under the general constraints.

By convention, we consider $D[i][i]=0$, which will be the base case. To write the recurrence formula, we need to take the file $e_{i_{max}}$ with maximum frequency $p_{i_{max}}, i_{max} \in [a+1,b]$. We go through every size index $k \in [1..m]$ such that $e_{i_{max}}$ can be padded to the size of $s_k$ and we update the channel-matrix according to Observation \ref{observation_1}, i.e. add $1$'s on $k$-th column if we can (taking into consideration the padding constraints) and complete the lines that have a fixed $1$ with $0$'s on the remaining entries. Then, we apply the recurrence on the rows which are not updated, i.e. from $a$ to $a^* \eqdef \max(a,\max_{i \in [1..n]}\{i |r_i < k\})$, and, respectively, from $b^* \eqdef \min(k,b)$ to $b$. Hence,
$$
    D[a][b]= p_{i_{max}}+
    \min_{k \in [l_{i_{max}}..r_{i_{max}}]}{(D[a][a^*]+D[b^*][b])}
$$
After applying the dynamic algorithm program with the aforementioned recurrence, we get the minimization of \eqref{sum_min} in $D[0][n]$, from which we can compute the minimal \theLeakage{}. If we want to recover the channel-matrix itself, in $D[a][b]$ we pass on the index $k$ for which the maximum happens, as an argument. In case of a tie, we choose the smallest index $k \in \set{1, \dots, n}$ in order to reduce average padding. Bandwidth minimization is further analyzed in Section~\ref{further_minimization}. Hence, we know in each sub-interval $[a,b]$ what we pad everything to, so the information is enough to recover the channel matrix. The implementation can be found in~\citeRepository{}.

In Figure~\ref{Figure_representation} we depict the channel-matrix of the files with sizes $S = \{1000, 1050, 1100, 1120, 1140\}$ and associated frequencies $\{22\%, 5\%, 23\%, 12\%, 18\%, 20\% \}$.
As shown in the visual representation of the padding-scheme in the right, we observe that, for both of the existing padded sizes, there are multiple files that are padded to the same element, making them indistinguishable for an attacker. Moreover, the \emph{blue} and \emph{red} bars on the graph indicate the frequencies of the files, respectively, the maximum frequency among the frequencies of the files padded to each specific size. These are used in order to highlight the terms of the sum \eqref{sum_min}.

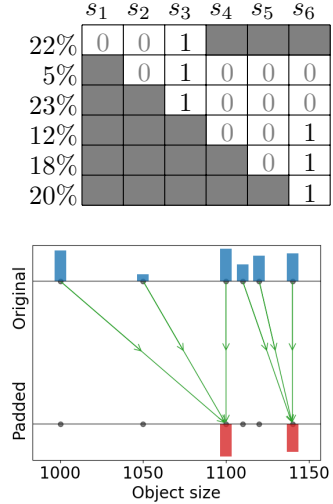
\begin{figure}[!ht]
  \begin{center}
    \raisebox{1em}{\def\Rcm{0.55cm}
\def\Ccm{0.4cm}
\begin{tikzpicture}[x={(0,-\Ccm)}, y={(\Rcm, 0)}, every node/.style={minimum width=\Rcm, minimum height=\Ccm, outer sep=0pt, anchor=north west}]
    \node[fill=gray, minimum width=1*\Rcm] at (1,0) {};
    \node[fill=gray, minimum width=2*\Rcm] at (2,0) {};
    \node[fill=gray, minimum width=3*\Rcm] at (3,0) {};
    \node[fill=gray, minimum width=4*\Rcm] at (4,0) {};
    \node[fill=gray, minimum width=5*\Rcm] at (5,0) {};
    \node[fill=gray, minimum width=3*\Rcm] at (0,3) {};
    \node at (0,-1.5) {$22\%$};
    \node at (1,-1.5) {$\phantom{0}5\%$};
    \node at (2,-1.5) {$23\%$};
    \node at (3,-1.5) {$12\%$};
    \node at (4,-1.5) {$18\%$};
    \node at (5,-1.5) {$20\%$};
    \node[xshift=-0.2*\Ccm] at (-1,0) {$s_1$};
    \node[xshift=-0.2*\Ccm] at (-1,1) {$s_2$};
    \node[xshift=-0.2*\Ccm] at (-1,2) {$s_3$};
    \node[xshift=-0.2*\Ccm] at (-1,3) {$s_4$};
    \node[xshift=-0.2*\Ccm] at (-1,4) {$s_5$};
    \node[xshift=-0.2*\Ccm] at (-1,5) {$s_6$};
    \node at (0,2) {$1$};
    \node at (1,2) {$1$};
    \node at (2,2) {$1$};
    \node at (3,5) {$1$};
    \node at (4,5) {$1$};
    \node at (5,5) {$1$};
    \node[color=gray] at (0,0) {$0$};
    \node[color=gray] at (0,1) {$0$};
    \node[color=gray] at (1,0) {$0$};
    \node[color=gray] at (1,1) {$0$};
    \node[color=gray] at (1,3) {$0$};
    \node[color=gray] at (1,4) {$0$};
    \node[color=gray] at (1,5) {$0$};
    \node[color=gray] at (2,1) {$0$};
    \node[color=gray] at (2,3) {$0$};
    \node[color=gray] at (2,4) {$0$};
    \node[color=gray] at (2,5) {$0$};
    \node[color=gray] at (3,3) {$0$};
    \node[color=gray] at (3,4) {$0$};
    \node[color=gray] at (4,3) {$0$};
    \node[color=gray] at (4,4) {$0$};
    \node[color=gray] at (5,4) {$0$};
    \draw[xstep=\Rcm, ystep=\Ccm,color=black,xshift=-0.5\pgflinewidth,yshift=0.5\pgflinewidth] (0,0) grid (6, 6);
  \end{tikzpicture}}
    \includegraphics[width=0.49\columnwidth]{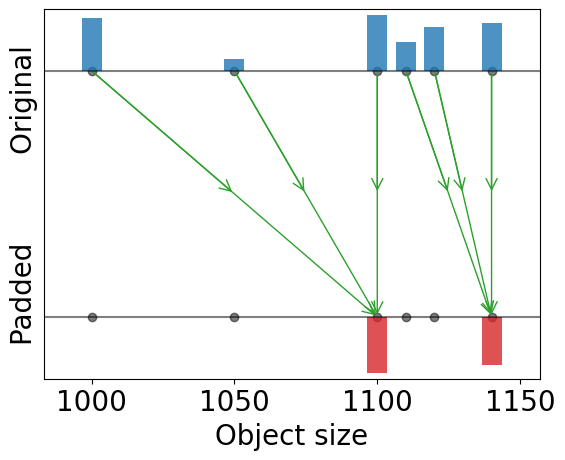}
  \end{center}
  \caption{Visualization of \PopRe{} for a dataset having the padding constraint $c=1.1$ and sizes  $\{1000, 1050$, $1100, 1110$, $1120, 1140\}$}
  \label{Figure_representation}
\end{figure}

\subsection{\PrpRe{}}

In this section, we treat the case of Per-Request-Padding and provide an algorithm for finding the probabilistic channel-matrix $P$ which minimizes the \theLeakage{}.
We will look at the joint distribution matrix $I$ with entries $I_{ij} = p_i \cdot p_{ij}, \forall i\leq n, j \leq m$, for which the row-wise sum $\sum_{j=1}^{m} I_{ij}$ is equal to $p_i, \forall i \in [1..n]$.

We proceed by finding iteratively, for each of the $m$ columns, starting from the last one, the \theLeakageOptimal{} manner of setting the entries of $I$ given the padding constraints. We define the \emph{optimal distribution of $p_i$ across the $i$-th row, $1 \leq i \leq n$} to be the way we fill in the entries $p_{i1}, \dots, p_{im}$ such as to obtain the minimum sum of the type \eqref{sum_min} and preserve the relation $p_{i1} + ... + p_{im} = p_i$.

The proof of our algorithm requires us to consider sub-problems in which the sequence $(p_i)_{1\leq i \leq n}$ is updated at each step of the algorithm, thus being different from the initial set of frequencies associated to each file. Hence, we rewrite the problem as a more general one in terms of a \emph{budget} sequence $\lr{b_i}_{1 \leq i\leq n}$ of length $n$ (initialized as $\lr{p_i}_{1 \leq i \leq n}$), which dictates the remaining value to be distributed across each row $i, \text{ for } i\in [1..n]$.
The general problem is stated below: 
\begin{adjustwidth}{1em}{1em}
  Given a non-negative budget sequence $\lr{b_i}_{i=1}^k$ of length $k\in [1..n]$, find a solution matrix $I_{k\times m}$ that minimizes Equation~\eqref{sum_min}, under the padding constraints for rows $i\in [1..k]$, namely the set $\{[l_1, r_1], \dots, [l_k, r_k]\}$ and $\sum_{j=1}^{m} I_{ij} = b_i$.
\end{adjustwidth}

We will design the algorithm to solve the general problem recursively by returning the matrix $I$ for the budget sequence $\{p_1, \dots, p_n\}$ with $n$ terms. The recurrence relationship can be described using the following observation that will be used when we create the probabilistic channel-matrix for the padding-scheme $f$:

\begin{observation} \label{observation_3}
    The solution $I_{k\times m}$ for a given $\lr{b_i}_{i=1}^k$ that minimizes Rényi leakage satisfies the recurrence relationship
    $$
    I_{ij} = 
    \begin{cases}
        b_i & \text{ if } j =m \text{ and } i \in [1..k], \abs{e_i}=s_m \\
        b_i - b_i^' & \text{ if } j =m \text{ and } i \in [1..k-1], \abs{e_i} \neq s_m, \\
         & m \in [l_i..r_i]\\
        I^'_{ij} & \text{ otherwise} \\
    \end{cases}
    $$
    where $I^'_{(k-t) \times (m-1)}$ is the solution to the same minimization problem for the sequence $\lr{b^'_i}_{i=1}^{k-t}$ of length $k-t$, $t = $ number of files from $E$ which can be padded to $s_m$, such that for any $i \in [1..k-t]$, it is defined as:
    $$
    b_i^' = 
    \begin{cases}
        \max(b_i - b_{t_{max}}, 0) & \text{ if } m \in [l_i..r_i] \text{ and }\\
        &  b_{t_{max}} = \max\{b_i | \abs{e_i}=s_m\}\\
        b_i & \text{ otherwise}
    \end{cases}
    $$
\end{observation}
\begin{proof}
    If there are no files among $\set{e_1, \dots, e_k}$ which can be padded to $s_m$, we set $t=0$ and solve the minimization problem for the same budget sequence and for the set of $m-1$ sizes $\set{s_1, \dots, s_{m-1}}$. 
    
    If there are files that can be padded to $s_m$, then due to the padding constraints, the element $e_i$ can only be padded to $s_m$, so the entry $I_{im}$ must necessarily be equal to $b_i, \text{ for all } i$ such that $\abs{e_i}=s_m$. Let us denote by $T = \set{k-t+1, \dots, k}$ the set of indices satisfying $\abs{e_i}=s_m, \forall i \in T$ and $b_{t_{max}} = \max\{b_i | i \in T\}$. Clearly, for every $i \in T$, $I_{ij} = 0, \forall j \in \set{1, \dots, k-1}$. On the $m$-th column of the matrix $I$, we have $\max_{i \in [1..k]} I_{im} \geq b_{t_{max}}$.
    
    In order to minimize the sum \eqref{sum_min} and taking into consideration that the maximum entry on column $m$ is at least $b_{t_{max}}$, we aim to distribute, for every $i$ such that $e_i$ can be padded to $s_m$ but $\abs{e_i} \neq s_m$, a quantity equal to $b_{t_{max}}$ (or, if $b_i < b_{t_{max}}$, then we distribute the whole $b_i$) on the entry $I_{im}$, such that we preserve the maximum on this last column to be $b_{t_{max}}$. This way, we can assure that, among the other columns, we'll have to distribute a smaller fraction of $b_i$, which means that the maximum on each column between $1$ and $m-1$ will decrease, and so will \eqref{sum_min}. 
    
    The problem reduces to find the optimal sub-matrix $I^'_{(k-t) \times (m-1)}$ to complete the first $k-t$ rows of $I$, and with the aforementioned observation, we can actually consider $I^'$ to be the solution given the updated sequence $(b^'_i)_{1 \leq i \leq k-t}$ which is defined, for every $i$ such that file $e_i$ that can be padded to $s_m$, as either $0, \text{ if } b_i \leq b_{t_{max}}$, or as $b_i - b_{t_{max}}, \text{ if } b_i \geq b_{t_{max}}$. When we reconstruct the matrix $I$, on the $m$-th column we will have the value $I^'_{im} + b_{t_{max}}$ or $I^'_{im} + b_i$ (depending whether $b_i$ is smaller, respectively larger, than $b_{t_{max}}$). 

    Now, let us show that, for the sub-matrix $I^'$, we have $0$'s on every entry of the $m$-th column. By definition, $I^'$ must be a \theLeakageOptimal{} solution for the updated sequence of $b^'_i$'s. Using Proposition \ref{proposition_1}, there exists a \theLeakageOptimal{} padding-scheme $f^'$ which maps $e_i, i \in [1..k-t] \rightarrow \{s_1, \dots, s_{k-t}\}$, for any set of files $\{e_1, \dots, e_{k-t}\}$ with the associated frequencies $\{b^'_1, \dots, b^'_{k-t}\}$. Consequently, for every $i \in [1..k-t], \mathbb{P}(f^'(e_i) = s_m)=0 \Rightarrow I^'_{im} = 0$.  
\end{proof}

Therefore, we have proved that the matrix $I$ can be recursively expressed using the sub-matrices obtained when we update the budget sequence accordingly, at each step decreasing by $1$ the number of columns and by at least $1$ the number of rows of the matrix returned from the algorithm, until we reduce a problem to finding the \theLeakageOptimal{} scheme for a budget sequence with a single element. Since we want to minimize \eqref{sum_min} in the case of $n$ files with frequencies $\{p_1, \dots, p_n\}$ and the associated set of sizes $\set{s_1, \dots, s_m}$, we proceed the induction on the number of rows and columns as described in Observation \ref{observation_3} and eventually fill in all the entries of the solution $I_{n \times m}$. If we want to recover the channel-matrix $P$ of conditional probabilities, we return the matrix with entries:
 $$
 P_{ij} = I_{ij} / p_i, \forall i\in [1..n], j \in [1..m].
 $$

This approach also handles a problem tackled in \cite{reed2021optimally}, in which the distribution of the files' frequencies is unknown. By supposing an uniform distribution of the files, our algorithm \PrpRe{} minimizes the Sibson mutual information of order infinity as introduced in \cite{issa2019operational}, hence solving the problem discussed in \cite{reed2021optimally}.

\section{Algorithms for further minimizations}\label{further_minimization}

We observe that there might be numerous padding-schemes that achieve minimal \theLeakage{}.
In this section we provide improvements to our algorithms in order to optimize average total padding and Shannon leakage after minimizing \theLeakage{}.
That is, we start with a \theLeakageOptimal{} padding scheme found with the algorithms of the previous section and apply several heuristics to reduce the average padding or the Shannon leakage, while preserving minimal \theLeakage{}.

\subsection{Average padding minimization}

In this section, we show improvements for the Per-Object-Padding and Per-Request-Padding cases in order to reduce the average padding while keeping the \theLeakage{} at its minimum.

\begin{enumerate}
\item
\textbf{POP(Per-Object-Padding)}

The heuristic for reducing bandwidth is already part of the design of algorithm \PopRe{}.
The heuristic consists of resolving ties by always choosing to pad the file with maximal frequency to the file with the least possible size.

Technically, this heuristic guarantees we find the padding-scheme with minimal bandwidth among those that are reachable using zero or more \emph{moves} (see Section~\ref{subsection_shannon_heuristic} for all details).
We verified this experimentally, with more than $10{,}000$ random generated tests using the code in~\citeRepository{}, as well as the fact that this heuristic is limited and does not find the minimal bandwidth among all padding functions with optimal \theLeakage{}.

\item
\textbf{PRP(Per-Request-Padding)}
Suppose $f$ is the padding scheme resulting from \PrpRe{} and $I$ is its joint distribution matrix. Let the list $C$ of maximums on each column, i.e. $C=\set{\max_{i\in [1..n]}I_{ij} | j \in [1..m]}$, where $C_j = \max_{i\in [1..n]}I_{ij}$ for every $j \in [1..m]$.

Define a \emph{move} to be a change in the matrix $I$ performed on two of the entries of the matrix on line $i, \text{ for some } i \in [1..n]$ such that $(I_{ia},I_{ib})$ becomes $(I_{ia}-\alpha, I_{ib}+\alpha)$ while keeping the entries of $I$ positive, i.e. $\alpha \leq I_{ia}$.

Now, we will describe an \emph{update} on line $i$, which will consist of a series of \emph{moves} and will return a new matrix $I^*$. We start with $I^*$ to be the matrix $I$, but with $0$'s on the $i$-th line. Since the sum on row $i$ is equal to $p_i$, we start with this quantity and go through the columns in order from $j=1$ to $j=m$. For each column, we set:
$$
I_{ij}=
\begin{cases}
C_j & \text{if } C_j + \sum_{k=1}^{j-1} I_{ik} \leq p_i \\
p_i-\sum_{k=1}^{j-1} I_{ik} & \text{ otherwise }
\end{cases}
$$
The algorithm ends by applying the update on all lines, regardless of their order. As for the case before, we will later discuss the performance of this algorithm. 

\let\underlineFirstLetter\underline 

In the next section we refer to this improvement of the algorithm which performs \underlineFirstLetter average \underlineFirstLetter padding \underlineFirstLetter minimization as \PrpReBa{}.

\end{enumerate}

\subsection{Shannon leakage minimization}\label{subsection_shannon_heuristic}

We aim to improve Shannon leakage while keeping the \theLeakage{} minimal.
For instance, in Figure~\ref{Figure_representation}, if we pad the file with frequency $10\%$ to size $96$ instead of $84$, then \theLeakage{} is preserved while Shannon leakage changes.

In this section, instead of working with the channel-matrix $P=(p_{ab})_{a\in [1..n], b \in [1..m]}$ of dimension $n \times m$, we will talk about the matrix $I=(p_{ab} \cdot p_a)_{a\in [1..n], b \in [1..m]}$. This is intuitive since minimizing the \theLeakage{} is equivalent to minimizing \eqref{sum_min}. However, the sum in \eqref{sum_min} is equal to $\sum_{j \in [1..m]} \max_{i \in [1..n]} (I_{ij})$, i.e. the sum of the maximal entries of the columns of matrix $I$.

We consider a \emph{move} to be a change of a padding-scheme $f$ to a different padding-scheme $f^*$ such that the two functions differ in one element only, and $f^*$ respects the padding constraints inherent to $f$. Our goal is to start from matrix $I$ and apply a series of \emph{moves} in order to eventually reach a matrix that preserves \theLeakage{} but has lower Shannon leakage.

We focus on keeping the \theLeakage{} constant after each single \emph{move}. This does not find the universal \theLeakageOptimal{} padding with the least Shannon leakage, a fact that was tested experimentally~\citeRepository{}, but in practice it performs very well (refer to Section~\ref{section_experiments} for details).

Starting from the initial matrix $I$, we want to get the matrix $I^*$ such that, at any step during the sequence of \emph{moves} applied, the following property holds for all columns $j \in [1..m]$:
$$
  \max_{i \in [1..n]} I_{ij}= \max_{i \in [1..n]} I^*_{ij}
$$

This property ensures that the \theLeakage{} stays the same due to the fact that the quantity in \ref{sum_min} remains constant. By restricting our algorithm to only perform \emph{moves} that respect the above condition, we define a list of positions where the file $e_i$ can be padded to, named $\texttt{poss[i]}$, for each $i$ in $[1..n]$:
$$
\texttt{poss[i]} \eqdef \{j \in [1..m] | j\in [l_i..r_i] \text{ and } \max_{a \in [1..n]} I_{aj} \geq p_i\}.
$$

Allowing to pad file $e_i$ only to files with index in $\texttt{poss[i]}$ will preserve the minimal \theLeakage{}. Now, the method to minimize the Shannon mutual information is rooted in the following underlying observation:

\begin{observation}\label{observation_2}
 Consider $f$ to be a \theLeakageOptimal{} padding-scheme with $P=(p_{ab})_{a \leq n, b \leq m}$ its channel-matrix and let $j\in [1..m]$ be the index of the column of the channel-matrix $(p_{ab})_{a\leq n, b \leq m}$ with the highest associated probability $\mathbb{P}(f(X) = s_j) =  {\sum_{a=1}^{n} p_a \cdot p_{aj} }$. Then, any \emph{move} of $1$'s to the $j$-th column gives a padding-scheme with less Shannon leakage.
 \end{observation}
\begin{proof}
 The Shannon mutual information is $\mathbb{I}(\size{X},f(X))=  \mathbb{H}(\size{X})-\mathbb{H}(\size{X} \mid f(X))=\mathbb{H}(f(X))-\mathbb{H}(f(X) \mid \size{X})$. Since we are still in the case POP, $f(X)$ is a deterministic function of $\size{X}$, so $\mathbb{H}(f(X) \mid \size{X})=0$. Hence, in order to minimize Shannon mutual information, it is sufficient to minimize: 
\begin{equation} \label{min_Shannon}
 \mathbb{H}(f(X))=- \sum_{j=1}^{m} \mathbb{P}(f(X) = s_j) \cdot \log_2 \mathbb{P}(f(X) = s_j)  
\end{equation}

Suppose there exists $i$ an index such that $p_{ij}=0$, and $i \in [l_j..r_j]$. Then, $ \exists k\neq j$ with $p_{ik}=0$. We want to show that \emph{moving} the $1$ from position $P_{ik}$ to $P_{ij}$ decreases the initial Shannon leakage. Consider the updated scheme $f^*$ resulting after the \emph{move} to column $j$, with the associated channel-matrix $P^*=(p_{ab}^*)_{a \leq n, b \leq m}$:
$$
p^*_{ab} = 
    \begin{cases}
        p_{ab} &  \text{ if } a \neq i \\
        1 & \text{ if } a = i \text{ and }b = j \\
        0 & \text{ otherwise }
    \end{cases}
$$
It is sufficient to prove that the sum \eqref{min_Shannon} is smaller for $f^*$ than for $f$. For simplicity, let us denote $\tilde{p}(x) = - x \cdot \log_2 x$, so \eqref{min_Shannon} can be written as $\sum_{b=1}^{m} \tilde{p}(g(b)) \text{ where for every } b \in [1..m], \text{ we define } g(b) := \mathbb{P}(f(X) = s_b) \text{ and }g^*(b) := \mathbb{P}(f^*(X) = s_b)$.

We see that: 
\begin{align*}
D = \sum_{b=1}^{m} \tilde{p}(g(b)) - \sum_{b=1}^{m} \tilde{p}(g^*(b)) =
 \tilde{p}(g(j)) + \tilde{p}(g(k)) -\\
 \tilde{p}(g^*(j)) - \tilde{p}(g^*(k)).
\end{align*}
But note that $g(j)>g(k)$, $g^*(j)>g(j)$, $g(k) > g^*(k)$ and, moreover, $g^*(j) - g(j) = p_t = g(k) - g^*(k)$. Using the properties of the concave function $\tilde{p}(x)$ on $[0,1]$, we get that $\tilde{p}(g(j)) + \tilde{p}(g(k)) \geq \tilde{p}(g^*(j)) +\tilde{p}(g^*(k))$. Hence $D \geq 0$, so the sum in \eqref{min_Shannon} is smaller for the padding-scheme $f^*$. 

This shows that the \emph{move} from position $P_{ik}$ to $P_{ij}$ decreases Shannon leakage.
\end{proof}

This observation implies that a padding-scheme with optimal \theLeakage{}, satisfying that no further \emph{move} can reduce the Shannon leakage, has a column, call it \emph{maximal} such that any $1 \text{ on line } i$ that could have been moved to that column, given \texttt{poss[i]}, was already \emph{moved} there. This corollary prompts us to leverage dynamic programming and divide the problem into sub-problems. 

We define $D[a][b]$ to be the partial Shannon leakage, looking at the columns from $a$ to $b$ on which we only place non-zero elements in lines $i \text{  such that } \texttt{poss[i]} \subset [a..b]$. So, any file which can be padded to a size $ s_t \text{ with } t \notin [a..b]$, will not be taken into account in the sub-problem $D[a][b]$. We will use the convention $D[j+1][j]=0$ since it acts on no columns, and we consider this to be a base case. With the observation above we deduce the following recurrence formula:
\begin{align*}
D[a][b]= \min_{j \in [a..b] } (D[a][i-1]+D[i+1][b]+ \\
\sum_{x\in \{i| j\in \texttt{poss}[i]\}}-x\cdot \log x)
\end{align*}

Note that this recurrence only keeps in memory the value of the Shannon leakage. If we want to recover the channel-matrix itself, we need to pass on as an argument the column which was \emph{maximal} in $[a..b]$ at each step. This information is enough to reconstruct the channel-matrix given the
\emph{maximal} property of the columns.

Below we give representations of the different outputs returned by the algorithms \PopReSh{} and \PopSh{}. By comparing Figure~\ref{Figure_representation} with Figure~\ref{Figure_why_renyi}, we highlight the reduction in Shannon leakage produced by \PopReSh{}. Moreover, in the case analyzed below, we observe that \PopSh{} gives a higher \theLeakage{} compared to \PopReSh{}. In the case of \PopSh{}, it is obvious for an attacker that the first padded size comes from the first file. An example of the proficiency of \PopReSh{} and \PopRe{} compared to \PopSh{} is that they don't have any file padded in an unique way, making it non-trivial for the attacker to guess the initial file in the example below.
\begin{figure}[!ht]
  \begin{center}
    \includegraphics[width=0.49\columnwidth]{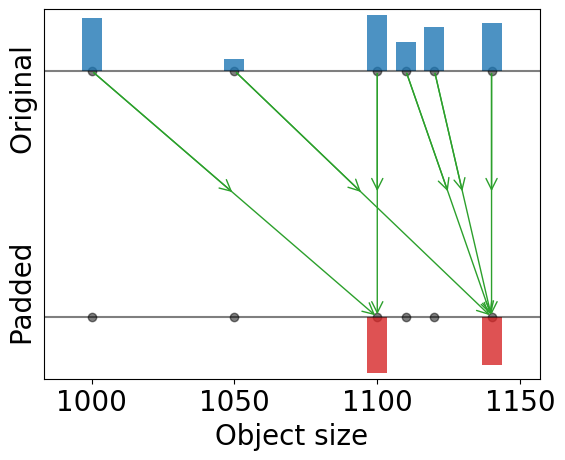}
    \includegraphics[width=0.49\columnwidth]{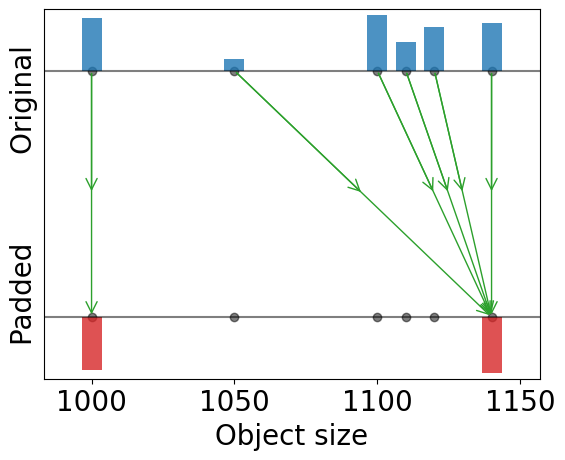}
  \end{center}
  \caption{Padding-schemes returned by \PopReSh{}(left) and \PopSh{} (right) in a fixed example; the value of the padding constraint is $c=1.1$}
  \label{Figure_why_renyi}
\end{figure}


\def\createPattern#1#2{
\pgfdeclarepatterninherentlycolored{horizontal lines #1} 
{\pgfpointorigin}{\pgfpoint{100pt}{4pt}}
{\pgfpoint{100pt}{4pt}}
{ \pgfsetfillcolor{#2!30!white} 
  \pgfpathrectangle{\pgfpointorigin}{\pgfpoint{100pt}{2.5pt}}
  \pgfusepath{fill}
  \pgfsetfillcolor{#2!25!white} 
  \pgfpathrectangle{\pgfpoint{0pt}{2pt}}{\pgfpoint{100pt}{2.5pt}}
  \pgfusepath{fill}
}
}
\createPattern{blue}{blue}
\createPattern{red}{red!80!black}
\createPattern{green}{green!75!black}
\createPattern{orange}{orange!75!black}
\createPattern{violet}{violet!75!black}
\createPattern{black}{black!75!black}

\pgfplotsset{
  adHocBarPlot/.style={
    ybar,
    legend cell align=left,
    bar width=0.12cm,
    height=0.7*0.9*\columnwidth,
    width=0.99*\columnwidth,
    grid=major, grid style={dashed},
    scaled x ticks=false,
    scaled y ticks=false,
    xlabel=$c$,
    enlarge x limits=0.15,
    enlarge y limits={value=0.75, upper},
    yticklabel style={
      /pgf/number format/fixed, 
      /pgf/number format/precision=3,
    },
    mark options={scale=0.4},
    legend pos=north east,
  },
  custom blue/.style={
    draw=blue, pattern=horizontal lines blue,
  },
  custom red/.style={
    draw=red!80!black, pattern=horizontal lines red,
  },
  custom green/.style={
    draw=green!75!black, pattern=horizontal lines green,
  },
  custom orange/.style={
    draw=orange!75!black, pattern=horizontal lines orange,
  },
  custom violet/.style={
    draw=violet!75!black, pattern=horizontal lines violet,
  },
}

\section{Experiments and Comparison}\label{section_experiments}

Several experiments were carried out for three distinct purposes, namely, (1) to test the correctness of the implementations against brute-force algorithms for small sized problems, (2) to corroborate the direct link between \theLeakage{} and the success rate of an attacker and (3) to compare the runtime, bandwidth and leakages of all the algorithms on a large real-world dataset.
The code of all the experiments is available in~\citeRepository{}.

\subsection{Brute-force tests}

All the algorithms developed, including our re-implementation of \PopSh{} to handle the server identity problem, were tested against brute-force implementations for small datasets (with at most $10$ elements).
That is, for each randomly generated test case of file sizes and frequencies, we explored (exhaustively) all the POP padding schemes satisfying the constraints, and chose among them, the ones that minimized \theLeakage{}, Shannon leakage or bandwidth, with the purpose of comparing them with the solutions returned by our algorithms.

We ran tens of thousands of experiments (code available in~\citeRepository{}), all corroborating that:
(a) among all POP schemes, \PopRe{} achieves minimal \theLeakage{} and \PopSh{} achieves minimal Shannon leakage;
(b) \PopReSh{} is always better than \PopRe{};
(c) \PrpRe{} leaks at most the \theLeakage{} of \PopRe{}.

\subsection{Attacker tests}

We simulated a real attacker who observes the padded size of a file transferred in the network and tries to guess the original file.
We assume that the attacker knows the joint distribution matrix $I$ and guesses, for a padded size $s_j, j \in [1..m]$, the file $e_{j^*}, j^* \in [1..n]$ which satisfies $p_{j^*} \cdot p_{j^*j} = \max_{a \in [1..n]} I_{aj}$.
If $A$ is the random variable, taking the value $0$ each time the attacker is wrong, and $1$ if the attacker is right, we see that the expected value of $A$ can be expressed in the following way:
$$\mathbb{E}[A]=\sum_{j=1}^{m} p_{j^*} \cdot p_{j^* j}.$$

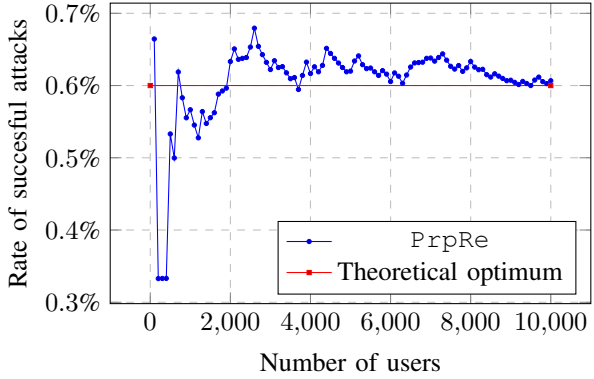
\begin{figure}[!ht]
  \begin{center}
    \begin{tikzpicture}[]
  \begin{axis}[
    height=0.7*0.9*\columnwidth,
    width=0.9*\columnwidth,
    grid=major, grid style={dashed},
    scaled x ticks=false,
    scaled y ticks=false,
    xlabel=Number of users,
    ylabel=Rate of succesful attacks,
    yticklabel style={
      /pgf/number format/fixed, 
      /pgf/number format/precision=3,
    },
    mark options={scale=0.4},
    yticklabel={\pgfmathprintnumber\tick\%},
    legend pos=south east,]
    \addplot table[
      x expr=100*(1+\coordindex),
      y=Y,
    ]{tikz-convergence.txt};
    \addlegendentry{\PrpRe{}};
    \addplot coordinates {(0,0.6) (10000,0.6) };
    \addlegendentry{Theoretical optimum};
  \end{axis}
\end{tikzpicture}
  \end{center}
  \caption{Convergence of the attacker's success}
  \label{Figure_convergence}
\end{figure}

This proves that the expected value of the success of the attacker, by the law of large numbers, should converge to $\mathbb{E}[A]$, which is equal to the value of \eqref{sum_min}.
This means that building a padding-scheme that minimizes the success of the attacker is indeed equivalent to finding the optimal \theLeakage{}, and this can be observed graphically in Figure~\ref{Figure_convergence}.

\subsection{NodeJS dataset tests}


The dataset of NodeJS packages was proposed in~\cite{reed2021optimally}.
This dataset consists of a list of $423{,}450$ javascript packages provided by NPM for browser and nodeJS applications, each with its associated file size and access frequency, as of August 2021.
These packages are used by most websites, which pull them from Content Delivery Networks such as jsDelivr, unpkg or cdnjs, and therefore, having information about pulled packages leaks some information about the website identity, since the specific packages used by each webpage can be though of a fingerprint.
For this security reason, and taking into account the large number of files, we used the NodeJS dataset to benchmark the algorithms.

We used two versions of the NodeJS dataset: the \emph{large} NodeJS dataset is the original dataset with 423,450 files, and the \emph{small} consists of only the 1000 most frequently accessed files.
The small NodeJS dataset allowed us to benchmark and compare the algorithms with large complexity, which timeout on the large dataset.

Firstly, we compared our own developed algorithms \PopRe{}, \PopReSh{}, \PrpRe{}, \PrpReBa{}, and the algorithm \PopSh{} from \cite{reed2021optimally}, which minimizes Shannon mutual information, against the Small NodeJS dataset. We present the following two plots which make the comparison between these algorithm in terms of both min-leakage and Shannon leakage:

\begin{figure}[!ht]
  \begin{center}
    \begin{tikzpicture}[]
  \begin{axis}[
    /pgfplots/adHocBarPlot,
    ylabel=\theLeakage{},
    enlarge y limits=0.25,
    ]
    
    \addplot[custom red,]
      table[x=c, y=renyi, ]{figures/nodejs-1000-data/PRP_Renyi_Bandwidth.txt};
    \addlegendentry{\PrpReBa{}};
    
    \addplot[custom orange,]
      table[x=c, y=renyi, ]{figures/nodejs-1000-data/PRP_Renyi_only.txt};
    \addlegendentry{\PrpRe{}};

    \addplot[custom green,]
      table[x=c, y=renyi, ]{figures/nodejs-1000-data/POP_Renyi_only.txt};
    \addlegendentry{\PopRe{}};
      
    \addplot[custom blue,]
      table[x=c, y=renyi, ]{figures/nodejs-1000-data/POP_Renyi_Shannon.txt};
    \addlegendentry{\PopReSh{}};

    \addplot[custom violet,]
      table[x=c, y=renyi, ]{figures/nodejs-1000-data/POP_Shannon_only.txt};
    \addlegendentry{\PopSh{}};
  \end{axis}
\end{tikzpicture}
  \end{center}
  \caption{\theLeakage{} using the small NodeJS dataset.}
  \label{Figure_npm_renyi}
\end{figure}

In Figure~\ref{Figure_npm_renyi}, we observe that, for the various values of the constant $c$ ranging from $1$ to $1.1$, the \theLeakage{} becomes significantly smaller for the \PrpReBa{} and \PrpRe{} in comparison with the rest of the algorithms presented.
This difference can be noticed from a slight change in $c$ from $1$, which describes that the file sizes are not padded at all, to the value $1.02$.
As the value of $c$ increases, allowing the files to be padded more, the \theLeakage{} is decreased by about $40 \%$ from $c=1$ to $c=1.1$.
We remark that, for each of the $6$ values of the padding constraint $c$, \PrpReBa{} and \PrpRe{} provide a smaller leakage than \PopRe{} and \PopReSh{}. However, the reduction of \PopRe{} is tiny compared to \PopSh{}, suggesting that the algorithm provided by Reed and Reiter~\cite{reed2021optimally} is comparable in terms of \theLeakage{} with our per-object-padding algorithm, however, our per-request-padding algorithm \PrpRe{} is more efficient to reduce information gain of an attacker. 

\begin{figure}[!ht]
  \begin{center}
    \begin{tikzpicture}[]
  \begin{axis}[
    /pgfplots/adHocBarPlot,
    ylabel=Shannon leakage,
    enlarge y limits=0.25,
    ]
    
    \addplot[custom red,]
      table[x=c, y=shannon, ]{figures/nodejs-1000-data/PRP_Renyi_Bandwidth.txt};
    \addlegendentry{\PrpReBa{}};
    
    \addplot[custom orange,]
      table[x=c, y=shannon, ]{figures/nodejs-1000-data/PRP_Renyi_only.txt};
    \addlegendentry{\PrpRe{}};

    \addplot[custom green,]
      table[x=c, y=shannon, ]{figures/nodejs-1000-data/POP_Renyi_only.txt};
    \addlegendentry{\PopRe{}};
      
    \addplot[custom blue,]
      table[x=c, y=shannon, ]{figures/nodejs-1000-data/POP_Renyi_Shannon.txt};
    \addlegendentry{\PopReSh{}};

    \addplot[custom violet,]
      table[x=c, y=shannon, ]{figures/nodejs-1000-data/POP_Shannon_only.txt};
    \addlegendentry{\PopSh{}};
  \end{axis}
\end{tikzpicture}
  \end{center}
  \caption{Shannon leakage using the small NodeJS dataset.}
  \label{Figure_npm_shannon}
\end{figure}
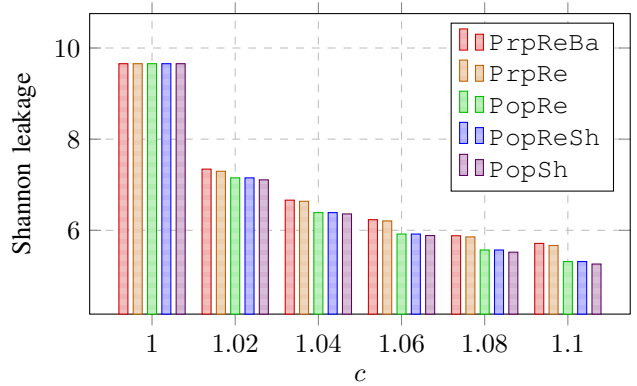

In Figure~\ref{Figure_npm_shannon}, we remark a reversed behavior than the one present in the previous graph.
More specifically, the Shannon leakage in the two per-request-padding case is greater than the one given by the three algorithms using per-object-padding.
Again, as $c$ increases by $10\%$, the Shannon leakage decreases by roughly $40\%$.
Nonetheless, the Shannon leakage returned by our improvement \PopReSh{} is similar overall with the Shannon leakage of \PopSh{}.
Consequently, our algorithm \PopReSh{} gives a better \theLeakage{} and a similar Shannon leakage compared to \PopSh{} developed by~\cite{reed2021optimally}.

\begin{figure}[!ht]
  \begin{center}
    \begin{tikzpicture}[]
  \begin{axis}[
    /pgfplots/adHocBarPlot,
    ylabel=Runtime (seconds),
    legend pos=north west,
    ymax=3,
    ]
    
    \addplot[custom red,]
      table[x=c, y=elapsed, ]{figures/latest-version/1000-PrpReBa.txt};
    \addlegendentry{\PrpReBa{}};
    
    \addplot[custom orange,]
      table[x=c, y=elapsed, ]{figures/latest-version/1000-PrpRe.txt};
    \addlegendentry{\PrpRe{}};

    \addplot[custom green,]
      table[x=c, y=elapsed, ]{figures/latest-version/1000-PopRe.txt};
    \addlegendentry{\PopRe{}};
      
    \addplot[custom blue,]
      table[x=c, y=elapsed, ]{figures/latest-version/1000-PopReSh.txt};
    \addlegendentry{\PopReSh{}};

    \addplot[custom violet,]
      table[x=c, y=elapsed, ]{figures/latest-version/1000-PopSh.txt};
    \addlegendentry{\PopSh{}};
  \end{axis}
\end{tikzpicture}
  \end{center}
  \caption{Small NodeJS runtime plot. The JIT compilation takes $7$ additional seconds.}
   \label{Figure_npm_runtime}
\end{figure}
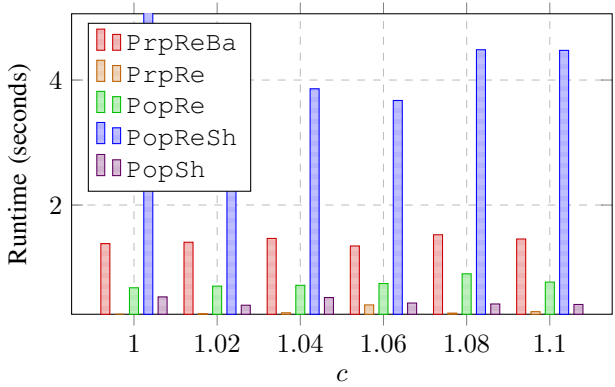

Figure~\ref{Figure_npm_runtime} depicts the runtime of the algorithms under analysis.
As predicted by the complexity analysis explained in section dedicated to algorithms and proofs, \PrpRe{} has the smallest runtime among the algorithms that we developed, and it is similar to \PopSh{} for the Small NodeJS dataset.
However, our algorithm remains more efficient in terms of the running time, as it can be observed for the cases $c=1.08$ and $c=1.1$.
As the value of $c$ increases, the execution time of \PopRe{} increases most significantly as it has complexity $O(n \, m^2)$, $m$ denoting the number of available sizes where to pad the files. Nonetheless, the same behavior occurs in the case of \PrpReBa{}.
The latter has a greater runtime than \PrpRe{}, difference which can be seen starting from $c=1.02$.
Their complexity is polynomially the same,  $O(n\,{\bar m})$ (refer to Table~\ref{table_complexities}), but in practice \PrpRe{} runs faster as it often fills up a row in $O(1)$ rather than $O(m)$ or $O({\bar m})$.
For the algorithm \PopReSh{}, the running time appears to be decreasing in the value of the padding constraint $c$.
This indicates that a higher $c$ allows \PopRe{} to give a result which has a Shannon leakage closer to the optimal, as seen in Figure~\ref{Figure_npm_shannon}, which diminishes the number of improvements that \PopReSh{} can act on and consequently its runtime.

\begin{figure}[!ht]
  \begin{center}
    \begin{tikzpicture}[]
  \begin{axis}[
    /pgfplots/adHocBarPlot,
    ylabel=Bandwidth increase,
    legend pos=north west,
    yticklabel={\pgfmathprintnumber\tick\%},
    ]
    
    \addplot[custom red,]
      table[x=c, y=bandwidth_percent, ]{figures/nodejs-1000-data/PRP_Renyi_Bandwidth.txt};
    \addlegendentry{\PrpReBa{}};
    
    \addplot[custom orange,]
      table[x=c, y=bandwidth_percent, ]{figures/nodejs-1000-data/PRP_Renyi_only.txt};
    \addlegendentry{\PrpRe{}};

    \addplot[custom green,]
      table[x=c, y=bandwidth_percent, ]{figures/nodejs-1000-data/POP_Renyi_only.txt};
    \addlegendentry{\PopRe{}};
      
    \addplot[custom blue,]
      table[x=c, y=bandwidth_percent, ]{figures/nodejs-1000-data/POP_Renyi_Shannon.txt};
    \addlegendentry{\PopReSh{}};

    \addplot[custom violet,]
      table[x=c, y=bandwidth_percent, ]{figures/nodejs-1000-data/POP_Shannon_only.txt};
    \addlegendentry{\PopSh{}};
  \end{axis}
\end{tikzpicture}
  \end{center}
  \caption{Bandwidth increase using the small NodeJS dataset; for reference, the average bandwidth use without padding is $52{,}487$}
   \label{Figure_npm_bandwidth_abs}
\end{figure}
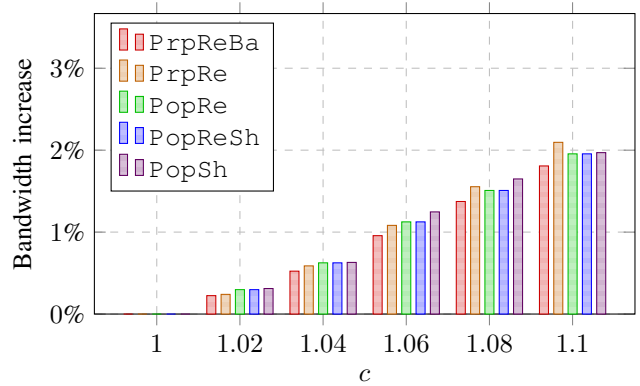

Next, we analyze the performance of our algorithms in terms of the bandwidth increase generated by the padding of the file sizes.
Figure~\ref{Figure_npm_bandwidth_abs} depicts the percentage increase in the bandwidth for the same values of $c$ previously analyzed.
The bandwidth increase is a non-decreasing function of the padding constraint, i.e. allowing the servers to pad more the files in order to find optimal \theLeakage{} also generates an increase in the bandwidth.
Hence, we conclude that, as the padding limitations are relaxed, the intuition behind the fact that it is better for reducing leakage to pad more the objects is reinforced.
As suggested by the graph, the algorithm \PrpReBa{} proves to be efficient in reducing the average padding over a network, as it always gives the smallest bandwidth increase out of the $5$ algorithms under analysis.
Conversely, \PopSh{} gives the highest average padding in all the cases except for $c=1.1$.

From now on, we analyze the large NodeJS dataset.
The only algorithms efficient enough to run on this data are \PrpReBa{}, \PrpRe{} and to a lesser extent \PopSh{}.
Once more, we will plot \theLeakage{}, Shannon leakage, bandwidth increase and runtime.

In Figure~\ref{Figure_large_npm_renyi}, as compared to Figure~\ref{Figure_npm_renyi}, we can see that a larger dataset allows for greater differences in \theLeakage{} between the optimal algorithm \PrpRe{} and \PopSh{}.
Moreover, the results of this figure corroborate the fact that \PrpRe{} and \PrpReBa{} give the same \theLeakage{}.

\begin{figure}[!ht]
  \begin{center}
    \begin{tikzpicture}[]
  \begin{axis}[
    /pgfplots/adHocBarPlot,
    ylabel=\theLeakage{},
    enlarge y limits=0.25,
    ]
    
    \addplot[custom red,]
      table[x=c, y=renyi, ]{figures/nodejs-large-data/PRP_Renyi_Bandwidth.txt};
    \addlegendentry{\PrpReBa{}};
    
    \addplot[custom orange,]
      table[x=c, y=renyi, ]{figures/nodejs-large-data/PRP_Renyi_only.txt};
    \addlegendentry{\PrpRe{}};

    \addplot[custom violet,]
      table[x=c, y=renyi, ]{figures/nodejs-large-data/POP_Shannon_only.txt};
    \addlegendentry{\PopSh{}};
  \end{axis}
\end{tikzpicture}
  \end{center}
  \caption{\TheLeakage{} using the entire NodeJS dataset}
  \label{Figure_large_npm_renyi}
\end{figure}

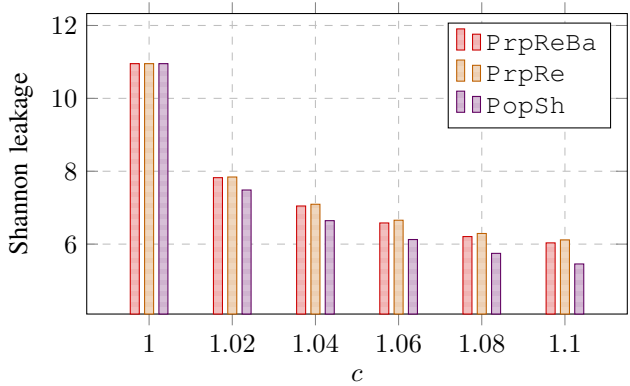
\begin{figure}[!ht]
  \begin{center}
    \begin{tikzpicture}[]
  \begin{axis}[
    /pgfplots/adHocBarPlot,
    ylabel=Shannon leakage,
    enlarge y limits=0.25,
    ]
    
    \addplot[custom red,]
      table[x=c, y=shannon, ]{figures/nodejs-large-data/PRP_Renyi_Bandwidth.txt};
    \addlegendentry{\PrpReBa{}};
    
    \addplot[custom orange,]
      table[x=c, y=shannon, ]{figures/nodejs-large-data/PRP_Renyi_only.txt};
    \addlegendentry{\PrpRe{}};

    \addplot[custom violet,]
      table[x=c, y=shannon, ]{figures/nodejs-large-data/POP_Shannon_only.txt};
    \addlegendentry{\PopSh{}};
  \end{axis}
\end{tikzpicture}
  \end{center}
  \caption{Shannon leakage using the entire NodeJS dataset}
  \label{Figure_large_npm_shannon}
\end{figure}

Figure~\ref{Figure_large_npm_shannon}, which analyses the Shannon leakage, depicts once again an opposite behavior compared to the one in Figure~\ref{Figure_large_npm_renyi}.
Naturally, \PopSh{} has a smaller Shannon leakage than our algorithms \PrpRe{} and \PrpReBa{}.
Once again, we notice that both the Shannon leakage and the \theLeakage{} follow a decreasing trend in the value of $c$, exactly as observed in the experiments performed for the small NodeJS dataset.

It can be observed in Figure~\ref{Figure_large_npm_bandwidth_abs} the proficiency of the optimization \PrpReBa{} compared to its initial variant \PrpRe{} and to \PopSh{}.
This makes \PrpReBa{} ideal for a web server who wants to limit the total expected average bandwidth on top of the individual bounds set by $c$.

Lastly, Figure~\ref{Figure_large_npm_runtime} highlights the scalability of the three algorithms, especially \PrpRe{} and \PrpReBa{}.
For all values of $c$ plotted in this graph, the runtime for \PrpRe{} is under $7$ seconds, which makes it the fastest algorithm among all of its competitors, while optimizing its measure of privacy, \theLeakage{}.
We can also see that \PrpReBa{} has a great running time compared to \PopSh{}, peaking at $c=1.1$ with around $3$ minutes versus $15$ minutes, which is a clear advantage.

\begin{figure}[!ht]
  \begin{center}
    \begin{tikzpicture}[]
  \begin{axis}[
    /pgfplots/adHocBarPlot,
    ylabel=Bandwidth increase,
    legend pos=north west,
    yticklabel={\pgfmathprintnumber\tick\%},
    ]
    
    \addplot[custom red,]
      table[x=c, y=bandwidth_percent, ]{figures/nodejs-large-data/PRP_Renyi_Bandwidth.txt};
    \addlegendentry{\PrpReBa{}};
    
    \addplot[custom orange,]
      table[x=c, y=bandwidth_percent, ]{figures/nodejs-large-data/PRP_Renyi_only.txt};
    \addlegendentry{\PrpRe{}};

    \addplot[custom violet,]
      table[x=c, y=bandwidth_percent, ]{figures/nodejs-large-data/POP_Shannon_only.txt};
    \addlegendentry{\PopSh{}};
  \end{axis}
\end{tikzpicture}
  \end{center}
  \caption{Bandwidth increase using the NodeJS dataset}
   \label{Figure_large_npm_bandwidth_abs}
\end{figure}
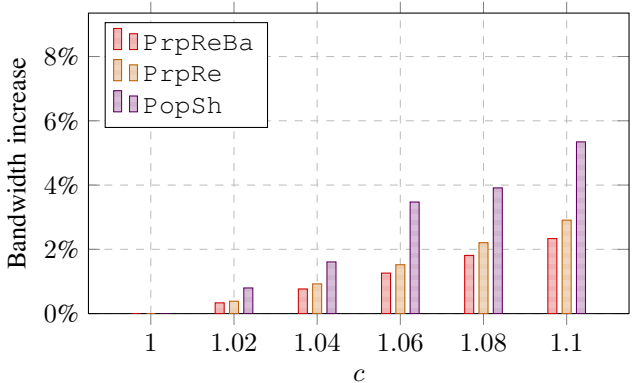

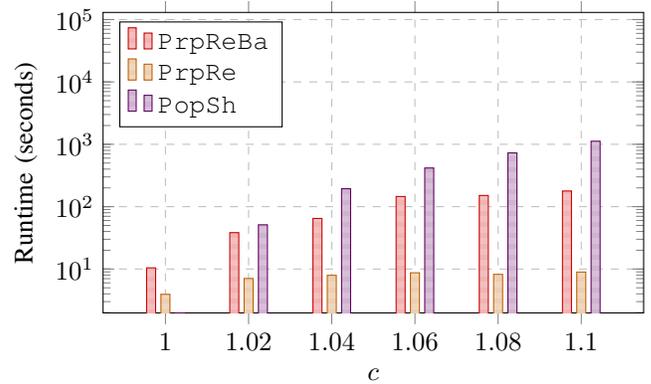
\begin{figure}[!ht]
  \begin{center}
    \begin{tikzpicture}[]
  \begin{axis}[
    /pgfplots/adHocBarPlot,
    ylabel=Runtime (seconds),
    legend pos=north west,
    ymode=log,
    ]
    
    \addplot[custom red,]
      table[x=c, y=elapsed, ]{figures/latest-version/423450-PrpReBa.txt};
    \addlegendentry{\PrpReBa{}};
    
    
    \addplot[custom orange,]
      table[x=c, y=elapsed, ]{figures/latest-version/423450-PrpRe.txt};
    \addlegendentry{\PrpRe{}};

    \addplot[custom violet,]
      table[x=c, y=elapsed, ]{figures/latest-version/423450-PopSh.txt};
    \addlegendentry{\PopSh{}};
  \end{axis}
\end{tikzpicture}
  \end{center}
  \caption{NodeJS runtime plot}
   \label{Figure_large_npm_runtime}
\end{figure}

\section{Future Work}

One natural improvement in our work would be to make our extended algorithms, namely \PopReSh{} and \PrpReBa{}, find the padding-scheme with optimal Shannon mutual information, respectively minimal average padding  among all the channels with minimal \theLeakage{}.
In order to better delve into the problem, let us consider a graph $G$ for each algorithm, such as the one in Figure~\ref{Figure_solid_dashed}, where the nodes represent all the \theLeakageOptimal{} padding-schemes.
We draw a \emph{dotted edge} between two vertices if we can apply exactly one \emph{move} (each algorithm has a different type of \emph{move}) to one of the padding-schemes and reach the other one.
On the other hand, we draw a \emph{straight edge} between two padding-schemes if we can apply a single \emph{move} to one function in order to reach the other one, and, furthermore, preserve the set of maximal elements on the columns, as in the solution for \PopReSh{}, or \PrpReBa{}.
Clearly, every two vertices united by a straight edge are also united by a dotted one.
Denote by $G_{\emph{\text{dotted edges}}}$ and $G_{\emph{\text{straight edges}}}$ the sub-graphs of $G$ containing solely dotted, respectively straight edges.
Our algorithms find the optimal vertex in the connected component of $G_{\emph{\text{dotted edges}}}$ which includes the node returned by \PopRe{}.
A future upgrade would be to find the optimal padding-scheme in the connected component of $G_{\emph{\text{straight edges}}}$, or in the whole graph $G$. 

For instance, if \PopRe{} returns one of the blue nodes in Figure~\ref{Figure_solid_dashed}, then \PopReSh{} will find the blue bold node, as it minimizes Shannon leakage in its connected component, thus it will fail to find the bold green node, which achieves the actual global minimum. 

\begin{figure}[!ht] 
\begin{center}
  \def\Rcm{0.8cm}
\def\Ccm{0.8cm}
\begin{tikzpicture}[x={(0,-\Ccm)}, y={(\Rcm, 0)}, every node/.style={minimum width=0.5*\Rcm, minimum height=0.5*\Ccm, outer sep=0pt, anchor=north west}]

  \tikzstyle{cyan} = [fill=cyan!50!white,shape=circle, draw=black];
  \tikzstyle{green} = [fill=green!50!white,shape=circle, draw=black];
  \tikzstyle{orange} = [fill=orange!50!white,shape=circle, draw=black];
  \tikzstyle{bold} = [very thick];

  \draw[->] (5.5, 0.7) -- (0.0, 0.7);
  \draw[->] (5.5, 0.7) -- (5.5, 7.0);
  \node[rotate=90] at (4.5, 0.0) {Shannon leakage};
  
  \node[orange, bold] (v0) at (4.2, 1.1) {};
  \node[orange] (v1) at (1.0, 1.0) {};
  \node[orange] (v2) at (1.2, 2.4) {};
  \node[orange] (v3) at (2.8, 1.2) {};
  \node[cyan] (v4) at (1.1, 3.6) {};
  \node[cyan, bold] (v5) at (4.0, 2.8) {};
  \node[cyan] (v6) at (2.3, 4.2) {};
  \node[cyan] (v7) at (0.8, 6.0) {};
  \node[green] (v8) at (2.8, 5.6) {};
  \node[green, bold] (v9) at (4.6, 4.8) {};

  \draw[      ] (v1) -- (v2);
  \draw[      ] (v2) -- (v3);
  \draw[dashed] (v3) -- (v5);
  \draw[dashed] (v2) -- (v4);
  \draw[      ] (v4) -- (v5);
  \draw[      ] (v4) -- (v6);
  \draw[      ] (v6) -- (v7);
  \draw[dashed] (v6) -- (v8);
  \draw[      ] (v8) -- (v9);
  \draw[      ] (v3) -- (v0);

\end{tikzpicture}
\end{center}
\caption{Graph containing the \theLeakage{} optimal padding-schemes}
\label{Figure_solid_dashed}
\end{figure}
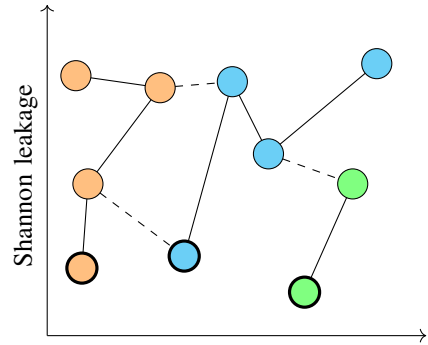

Secondly, we prioritized the minimization of \theLeakage{} over the measure of average padding. Future work can be directed in finding an optimal trade-off between the two, not necessarily giving absolute priority to the bandwidth of each costumer using the constant $c$. 

Lastly, our paper and \cite{reed2021optimally} provided algorithms to minimize Rényi-min and Shannon leakage.
As mentioned in the subsection Problem formalization, we could also try to minimize $\mathbb{I}_{\alpha}$, the \theLeakage{}.
However, thanks to the applicability of the notion of \theLeakage{} to measure the success of a real adversary, we believe that \theLeakage{} is the most adequate measure of leakage for this problem.

\section{Conclusion}

In this paper we used \theLeakage{} as the main measure for privacy, arguing that it is the most suitable for simulating a real-life attacker.
We built two types of algorithms that minimize \theLeakage{}, treating not only the per-object-padding case in which the files are padded in a deterministic way, but also the per-request-padding, in which each single file is padded to every available size with a given probability.

We provided several improvements that lowered the average padding and the Shannon leakage.
We want to highlight the algorithms \PrpRe{} and \PrpReBa{}, which have proven a better running time compared to all of their competitors.
Nonetheless, the latter one also has an impressive bandwidth advantage.
We analyzed all of our algorithms numerically, using many artificial examples and also the NodeJS dataset.

Furthermore, we described a method for designing padding-schemes for multiple servers that protect the identities of both the files and the servers.
The servers should simply agree on a fixed standard set of possible output sizes and run our algorithms.

\bibliographystyle{alpha}
\bibliography{bibliography}

\end{document}